\documentclass[9pt, onecolumn]{article}
\newcommand{\CLASSINPUTtoptextmargin}{0.7in}

\usepackage[utf8]{inputenc}
\usepackage{amsthm}
\usepackage{amsmath}
\usepackage{amsfonts}
\usepackage{amssymb}
\usepackage{xcolor}
\usepackage{algorithm}
\usepackage{algpseudocode}
\usepackage{graphicx}
\usepackage{subfigure}
\usepackage{caption}
\usepackage{subfig}

\makeatletter
\newcommand*{\rom}[1]{\expandafter\@slowromancap\romannumeral #1@}
\makeatother
\newtheorem{theorem}{Theorem}

\newcommand{\firstAuthor}       {Raj~Thilak~Rajan} %
\newcommand{\theTitle}{On the choice of reference
in offset calibration}

\title{\theTitle}

\author{\firstAuthor}

\begin{document}



\newcommand{\incgrl}{\includegraphics[scale=0.7]}
\newcommand{\incgrs}{\includegraphics[scale=0.4]}
\def\percentile{^{^o/_{oo}}}
\newcommand{\mbf}{\mathbf}
\newcommand{\graphwidth}{7.2cm}
\newcommand{\define}{\equiv}
\newcommand{\half}{\frac{1}{2}}
\newcommand{\invN}{\frac{1}{N}}
\newcommand{\logtwo}{\,^2\hspace{-0.05cm}\log}
\newcommand{\logten}{\,^{10}\hspace{-0.05cm}\log}
\newcommand{\bzero}{\mbf{0}}
\newcommand{\nsp}{\vspace{-0.2cm}}
\def\b1{{\mathbf 1}}
\newcommand{\defeq}{\stackrel{\mathrm{def}}{=}}
\def\bell{{\ensuremath{\mathbf{\ell}}}}
\newcommand{\bRfilt}{\check{\ensuremath{\mathbf{R}}}}
\newcommand{\bxfilt}{\check{\ensuremath{\mathbf{x}}}}

\newcommand{\definesymbollistoffsets}{\betab \hspace{0.5cm}\= \hspace{3.0cm} \= \hspace{1cm} \= \kill}
\newcommand{\pbA}{\parbox[t]{2.3cm}}
\newcommand{\pbB}{\parbox[t]{10cm}}

\def\beB{{\mbox{\boldmath{$B$}}}}
\def\beE{{\mbox{\boldmath{$E$}}}}
\def\beJ{{\mbox{\boldmath{$J$}}}}
\def\beR{{\mbox{\boldmath{$R$}}}}
\def\beW{{\mbox{\boldmath{$W$}}}}

\newcommand{\PDF}{probability density function }
\newcommand{\PDFs}{probability density functions }

%
\newcommand{\balpha}{\boldsymbol\alpha}
\newcommand{\bbeta}{\boldsymbol\beta}
\newcommand{\bomega}{\boldsymbol\omega}
\newcommand{\bphi}{\boldsymbol\phi}
\def\bbell{{\mbox{\boldmath{$\bell$}}}}
\def\bDelta{{\mbox{\boldmath{$\Delta$}}}}
\def\bepsilon{{\mbox{\boldmath{$\epsilon$}}}}
\def\bgamma{{\mbox{\boldmath{$\gamma$}}}}
\def\bGamma{{\mbox{\boldmath{$\Gamma$}}}}
\def\bLambda{{\mbox{\boldmath{$\Lambda$}}}}
\def\blambda{{\mbox{\boldmath{$\lambda$}}}}
\def\bmu{{\mbox{\boldmath{$\mu$}}}}
\def\bphi{{\mbox{\boldmath{$\phi$}}}}
\def\bPhi{{\mbox{\boldmath{$\Phi$}}}}
\def\bOmega{{\mbox{\boldmath{$\Omega$}}}}
 \def\bPsi{{\mbox{\boldmath{$ \Psi$}}}}
 \def\bXi{{\mbox{\boldmath{$ \Xi$}}}}
 \def\bxi{{\mbox{\boldmath{$ \xi$}}}}
\def\bPsih{{\mbox{\boldmath{$ \Psi$}}}}
\newcommand{\bPhic}{{\ensuremath{\mathbf{\bar{\Phi}}}}}
\def\btau{{\mbox{\boldmath{$\tau$}}}}
\def\bsigma{{\mbox{\boldmath{$\sigma$}}}}
\def\bSigma{{\mbox{\boldmath{$\Sigma$}}}}
\def\bthetaa{\hat{{\mbox{\boldmath{$\theta$}}}}}
\newcommand{\bLambdaa}{{\ensuremath{\mathbf{\hat{\Lambda}}}}}

\def\brho{{\mbox{\boldmath{$\rho$}}}}
\def\bdelta{{\mbox{\boldmath{$\delta$}}}}

\newcommand{\cA}{\ensuremath{\mathcal{A}}}
\newcommand{\cB}{\ensuremath{\mathcal{B}}}
\newcommand{\cC}{\ensuremath{\mathcal{C}}}
\newcommand{\cD}{\ensuremath{\mathcal{D}}}
\newcommand{\cE}{\ensuremath{\mathcal{E}}}
\newcommand{\cF}{\ensuremath{\mathcal{F}}}
\newcommand{\cG}{\ensuremath{\mathcal{G}}}
\newcommand{\cH}{\ensuremath{\mathcal{H}}}
\newcommand{\cI}{\ensuremath{\mathcal{I}}}
\newcommand{\cJ}{\ensuremath{\mathcal{J}}}
\newcommand{\cK}{\ensuremath{\mathcal{K}}}
\newcommand{\cL}{\ensuremath{\mathcal{L}}}
\newcommand{\cM}{\ensuremath{\mathcal{M}}}
\newcommand{\cN}{\ensuremath{\mathcal{N}}}
\newcommand{\cO}{\ensuremath{\mathcal{O}}}
\newcommand{\cP}{\ensuremath{\mathcal{P}}}
\newcommand{\bcR}{\boldmath{\ensuremath{\mathcal{R}}}}
\newcommand{\cR}{\ensuremath{\mathcal{R}}}
\newcommand{\cS}{\ensuremath{\mathcal{S}}}
\newcommand{\cT}{\ensuremath{\mathcal{T}}}
\newcommand{\cU}{\ensuremath{\mathcal{U}}}
\newcommand{\cV}{\ensuremath{\mathcal{V}}}
\newcommand{\cX}{\ensuremath{\mathcal{X}}}
\newcommand{\cY}{\ensuremath{\mathcal{Y}}}
\newcommand{\cW}{\ensuremath{\mathcal{W}}}
\newcommand{\bcW}{{\ensuremath{\mathbf{\mathcal{W}}}}}
\newcommand{\cZ}{\ensuremath{\mathcal{Z}}}

\newcommand{\cs}{\ensuremath{\mathcal{s}}}
\newcommand{\cf}{\ensuremath{\mathcal{f}}}
\newcommand{\cg}{\ensuremath{\mathcal{g}}}
\newcommand{\ch}{\ensuremath{\mathcal{h}}}

\newcommand{\uN}{\ensuremath{\underline{N}}}
\newcommand{\bXu}{\ensuremath{\underline{\bX}}}
\newcommand{\bYu}{\ensuremath{\underline{\bY}}}
\newcommand{\bZu}{\ensuremath{\underline{\bZ}}}
\newcommand{\bAu}{\ensuremath{\underline{\bA}}}
\newcommand{\bBu}{\ensuremath{\underline{\bB}}}
\newcommand{\bCu}{\ensuremath{\underline{\bC}}}
\newcommand{\bUu}{\ensuremath{\underline{\bU}}}
\newcommand{\bFu}{\ensuremath{\underline{\bF}}}
\newcommand{\bGu}{\ensuremath{\underline{\bG}}}
\newcommand{\bJu}{{\ensuremath{\underline{\bJ}}}}
\newcommand{\bHu}{{\ensuremath{\underline{\bH}}}}
\newcommand{\bWu}{{\ensuremath{\underline{\bW}}}}

\newcommand{\brhou}{\ensuremath{\underline{\brho}}}
\newcommand{\bphiu}{\ensuremath{\underline{\bphi}}}
\newcommand{\bChiu}{\ensuremath{\underline{\bChi}}}
\newcommand{\bSigmau}{\ensuremath{\underline{\bSigma}}}

\newcommand{\baa}{{\ensuremath{\mathbf{\hat{a}}}}}
\newcommand{\bda}{{\ensuremath{\mathbf{\hat{d}}}}}
\newcommand{\bDela}{{\ensuremath{\mathbf{\hat{\Delta}}}}}
\newcommand{\bga}{{\ensuremath{\mathbf{\hat{g}}}}}
\newcommand{\bua}{{\ensuremath{\mathbf{\hat{u}}}}}
\newcommand{\bAa}{{\ensuremath{\mathbf{\widehat{A}}}}}
\newcommand{\bBa}{{\ensuremath{\mathbf{\widehat{B}}}}}
\newcommand{\bDa}{{\ensuremath{\mathbf{\widehat{D}}}}}
\newcommand{\bGa}{{\ensuremath{\mathbf{\widehat{G}}}}}
\newcommand{\bPa}{{\ensuremath{\mathbf{\widehat{P}}}}}
\newcommand{\bQa}{{\ensuremath{\mathbf{\widehat{Q}}}}}
\newcommand{\bRa}{{\ensuremath{\mathbf{\widehat{R}}}}}
\newcommand{\bUa}{{\ensuremath{\mathbf{\widehat{U}}}}}

\newcommand{\xc}{\bar{x}}
\newcommand{\bxc}{{\ensuremath{\mathbf{\bar{x}}}}}
\newcommand{\bgc}{{\ensuremath{\mathbf{\bar{g}}}}}
\newcommand{\bGc}{{\ensuremath{\mathbf{\bar{G}}}}}
\newcommand{\bWc}{{\ensuremath{\mathbf{\bar{W}}}}}
\newcommand{\bRc}{{\ensuremath{\mathbf{\bar{R}}}}}

\newcommand{\bVn}{\bV^{\perp}}
\newcommand{\bUn}{\bU^{\perp}}
\newcommand{\rE}{{\rm E}}

\newcommand{\argmin}{\mbox{\rm arg}\min}
\newcommand{\abs}{\mbox{\rm abs}}
\newcommand{\blockdiag}{\mbox{\rm blockdiag}}
\newcommand{\col}{\mbox{\rm col}}
\newcommand{\cond}{\mbox{\rm cond}}
\newcommand{\cov}{\mbox{\rm cov}}
\newcommand{\CRB}{\mbox{\rm CRB}}
\newcommand{\cum}{\mbox{\rm cum}}
\newcommand{\diag}{\mbox{\rm diag}}
\newcommand{\bdiag}{\mbox{\rm bdiag}}
\newcommand{\eig}{\mbox{\rm eig}}
\renewcommand{\Im}{\mbox{\rm Im}}
\newcommand{\imag}{\mbox{\rm imag}}
\newcommand{\sinc}{\mbox{\rm sinc}}
\newcommand{\INSR}{\mbox{\rm INSR}}
\newcommand{\order}{\mbox{\rm \cO}}
\newcommand{\rank}{\mbox{\rm rank}}
\renewcommand{\Re}{\mbox{\rm Re}}
\newcommand{\range}{\mbox{\rm ran}}
\newcommand{\real}{\mbox{\rm real}}
\newcommand{\row}{\mbox{\rm row}}
\newcommand{\sign}{\mbox{\rm sign}}
\newcommand{\INR}{\mbox{\rm INR}}
\newcommand{\SINR}{\mbox{\rm SINR}}
\newcommand{\SNIR}{\mbox{\rm SNIR}}
\newcommand{\SIR}{\mbox{\rm SIR}}
\newcommand{\SNR}{\mbox{\rm SNR}}
\newcommand{\spann}{\mbox{\rm span}}
\newcommand{\tr}{{\ensuremath{\mbox{\rm tr}}}}
\newcommand{\unvec}{\mbox{\rm unvec}}
\newcommand{\unvect}{{\vect}^{-1}}
\newcommand{\var}{\mbox{\rm var}}
\newcommand{\vect}{\mbox{\rm vec}}
\newcommand{\vecdiag}{\mbox{\rm vecdiag}}

\def\cf{{cf.\ }}
\def\viz{{viz.\ }}
\def\etal{{et\ al.\ }}
\def\eg{{e.g.,\ }}
\def\ie{{i.e.,\ }}

\newcommand{\betab}{\begin{tabbing}}
\newcommand{\entab}{\end{tabbing}}
\newcommand{\beitem}{\begin{itemize}}
\newcommand{\enitem}{\end{itemize}}
\newcommand{\bea}{\begin{array}}
\newcommand{\ena}{\end{array}}
\newcommand{\beq}{\begin{equation}}
\newcommand{\enq}{\end{equation}}
\newcommand{\beqa}{\begin{eqnarray}}
\newcommand{\enqa}{\end{eqnarray}}
\newcommand{\beqan}{\begin{eqnarray*}}
\newcommand{\enqan}{\end{eqnarray*}}
\newcommand{\beenum}{\begin{enumerate}}
\newcommand{\enenum}{\end{enumerate}}
\newcommand{\DL}{\begin{dashlist}}
\newcommand{\DLE}{\end{dashlist}}

\newcommand{\ba}{{\ensuremath{\mathbf{a}}}}
\newcommand{\bb}{{\ensuremath{\mathbf{b}}}}
\newcommand{\bc}{{\ensuremath{\mathbf{c}}}}
\newcommand{\bd}{{\ensuremath{\mathbf{d}}}}
\newcommand{\be}{{\ensuremath{\mathbf{e}}}}
\newcommand{\bff}{{\ensuremath{\mathbf{f}}}}
\newcommand{\bg}{{\ensuremath{\mathbf{g}}}}
\newcommand{\bh}{{\ensuremath{\mathbf{h}}}}
\newcommand{\bk}{{\ensuremath{\mathbf{k}}}}
\newcommand{\bl}{{\ensuremath{\mathbf{l}}}}
\newcommand{\bm}{{\ensuremath{\mathbf{m}}}}
\newcommand{\bn}{{\ensuremath{\mathbf{n}}}}
\newcommand{\bp}{{\ensuremath{\mathbf{p}}}}
\newcommand{\bq}{{\ensuremath{\mathbf{q}}}}
\newcommand{\br}{{\ensuremath{\mathbf{r}}}}
\newcommand{\bs}{{\ensuremath{\mathbf{s}}}}
\newcommand{\bt}{{\ensuremath{\mathbf{t}}}}
\newcommand{\bu}{{\ensuremath{\mathbf{u}}}}
\newcommand{\bv}{{\ensuremath{\mathbf{v}}}}
\newcommand{\bw}{{\ensuremath{\mathbf{w}}}}
\newcommand{\bx}{{\ensuremath{\mathbf{x}}}}
\newcommand{\by}{{\ensuremath{\mathbf{y}}}}
\newcommand{\bz}{{\ensuremath{\mathbf{z}}}}

\newcommand{\brdot}{{\ensuremath{\mathbf{\dot{\br}}}}}
\newcommand{\brrdot}{{\ensuremath{\mathbf{\ddot{\br}}}}}

\newcommand{\aX}{{\ensuremath{\mathbf{\breve{X}}}}}
\newcommand{\ad}{{\ensuremath{\mathbf{\breve{d}}}}}
\newcommand{\tX}{{\ensuremath{\mathbf{\tilde{X}}}}}
\newcommand{\td}{{\ensuremath{\mathbf{\tilde{d}}}}}

\newcommand{\bA}{{\ensuremath{\mathbf{A}}}}
\newcommand{\bB}{{\ensuremath{\mathbf{B}}}}
\newcommand{\bC}{{\ensuremath{\mathbf{C}}}}
\newcommand{\bD}{{\ensuremath{\mathbf{D}}}}
\newcommand{\bE}{{\ensuremath{\mathbf{E}}}}
\newcommand{\bF}{{\ensuremath{\mathbf{F}}}}
\newcommand{\bG}{{\ensuremath{\mathbf{G}}}}
\newcommand{\bH}{{\ensuremath{\mathbf{H}}}}
\newcommand{\bI}{{\ensuremath{\mathbf{I}}}}
\newcommand{\bJ}{{\ensuremath{\mathbf{J}}}}
\newcommand{\bK}{{\ensuremath{\mathbf{K}}}}
\newcommand{\bL}{{\ensuremath{\mathbf{L}}}}
\newcommand{\bM}{{\ensuremath{\mathbf{M}}}}
\newcommand{\bN}{{\ensuremath{\mathbf{N}}}}
\newcommand{\bO}{{\ensuremath{\mathbf{O}}}}
\newcommand{\bP}{{\ensuremath{\mathbf{P}}}}
\newcommand{\bQ}{{\ensuremath{\mathbf{Q}}}}
\newcommand{\bR}{{\ensuremath{\mathbf{R}}}}
\newcommand{\bS}{{\ensuremath{\mathbf{S}}}}
\newcommand{\bT}{{\ensuremath{\mathbf{T}}}}
\newcommand{\bU}{{\ensuremath{\mathbf{U}}}}
\newcommand{\bV}{{\ensuremath{\mathbf{V}}}}
\newcommand{\bW}{{\ensuremath{\mathbf{W}}}}
\newcommand{\bX}{{\ensuremath{\mathbf{X}}}}
\newcommand{\bY}{{\ensuremath{\mathbf{Y}}}}
\newcommand{\bZ}{{\ensuremath{\mathbf{Z}}}}

\newcommand{\xt}{\widetilde{\ensuremath{x}}}
\newcommand{\bxt}{\widetilde{\ensuremath{\bx}}}
\newcommand{\et}{\widetilde{\ensuremath{e}}}
\newcommand{\bet}{\widetilde{\ensuremath{\be}}}

\newcommand{\bAt}{\widetilde{\ensuremath{\mathbf{A}}}}
\newcommand{\bBt}{\widetilde{\ensuremath{\mathbf{B}}}}
\newcommand{\bCt}{\widetilde{\ensuremath{\mathbf{C}}}}
\newcommand{\bDt}{\widetilde{\ensuremath{\mathbf{D}}}}
\newcommand{\bEt}{\widetilde{\ensuremath{\mathbf{E}}}}
\newcommand{\bGt}{\widetilde{\ensuremath{\mathbf{G}}}}
\newcommand{\bJt}{\widetilde{\ensuremath{\mathbf{J}}}}
\newcommand{\bnt}{\widetilde{\ensuremath{\mathbf{n}}}}
\newcommand{\bRt}{\widetilde{\ensuremath{\mathbf{R}}}}
\newcommand{\bXt}{\widetilde{\ensuremath{\mathbf{X}}}}
\newcommand{\bYt}{\widetilde{\ensuremath{\mathbf{Y}}}}
\newcommand{\bZt}{\widetilde{\ensuremath{\mathbf{Z}}}}
\newcommand{\bVt}{\widetilde{\ensuremath{\mathbf{V}}}}
\newcommand{\bWt}{\widetilde{\ensuremath{\mathbf{W}}}}

\newcommand{\bdA}{\bA^\dagger}

\newcommand{\stokesI}{\ensuremath{\begin{pmatrix} 1 & 0 \\ 0 &  1\end{pmatrix}}}
\newcommand{\stokesQ}{\ensuremath{\begin{pmatrix} 1 & 0 \\ 0 & -1\end{pmatrix}}}
\newcommand{\stokesU}{\ensuremath{\begin{pmatrix} 0 & 1 \\ 1 &
0\end{pmatrix}}}
\newcommand{\stokesV}{\ensuremath{\begin{pmatrix} 0 & -\imath \\ \imath & 0\end{pmatrix}}}

\def\bSigma{{\mbox{\boldmath{$\Sigma$}}}}
\def\bSum{{\mbox{\boldmath{$\sum$}}}}
\def\cSigma{\emph{\mbox{\boldmath{$\Sigma$}}}}
\def\bTheta{{\mbox{\boldmath{$\Theta$}}}}
\def\btheta{{\mbox{\boldmath{$\theta$}}}}
\def\bOmega{{\mbox{\boldmath{$\Omega$}}}}
\def\btau{{\mbox{\boldmath{$\tau$}}}}
\def\bnu{{\mbox{\boldmath{$\nu$}}}}
\def\bpsi{{\mbox{\boldmath{$\psi$}}}}
\def\bupsilon{{\mbox{\boldmath{$\upsilon$}}}}
\def\biota{{\mbox{\boldmath{$\iota$}}}}
\def\bEta{{\mbox{\boldmath{$\eta$}}}}
\def\bZeta{{\mbox{\boldmath{$\zeta$}}}}
\def\brho{{\mbox{\boldmath{$\rho$}}}}
\def\bxi{{\mbox{\boldmath{$\xi$}}}}
\def\bXi{{\mbox{\boldmath{$\Xi$}}}}
\def\tbJ{{\mbox{\boldmath{$\tilde{\bJ}$}}}}
\def\abJ{{\mbox{\boldmath{$\acute{\bJ}$}}}}
\def\wrt{{w.r.t.\ }}
\def\iid{{i.i.d.\ }}
\newcommand{\norm}[1]{\lVert#1\rVert}
\newcommand{\normTwo}[1]{\lVert#1\rVert_2}

\def\ur{{\underline{r}}}
\def\bur{{\underline{\br}}}

\newcommand{\cd}{\ensuremath{\mathcal{d}}}

\newcommand{\cramer}{cram\'{e}r}
\newcommand{\Cramer}{Cram\'{e}r}

\newcommand{\barN}{{\ensuremath{\bar{N}}}}
\newcommand{\barP}{{\ensuremath{\bar{P}}}}

\newcommand{\bdu}{{\ensuremath{\underline{\bd}}}}
\newcommand{\bxu}{{\ensuremath{\underline{\bx}}}}
\newcommand{\byu}{{\ensuremath{\underline{\by}}}}

\newcommand{\thetau}{{\ensuremath{\underline{\theta}}}}
\newcommand{\bthetau}{{\ensuremath{\underline{\btheta}}}}
\newcommand{\trace}{{\ensuremath{\text{Tr}}}}

\def\boldeta{{\mbox{\boldmath{$\eta$}}}}

\newcommand{\Ver}{{\ensuremath{\mathcal{V}}}}
\newcommand{\Edg}{{\ensuremath{\mathcal{E}}}}

\newcommand{\bdX}   {{\ensuremath{\dot  {\bX}}}}
\newcommand{\bddX}  {{\ensuremath{\ddot {\bX}}}}
\newcommand{\bdddX} {{\ensuremath{\dddot{\bX}}}}

\newcommand{\bdx}   {{\ensuremath{\dot  {\bx}}}}
\newcommand{\bddx}  {{\ensuremath{\ddot {\bx}}}}
\newcommand{\bdddx} {{\ensuremath{\dddot{\bx}}}}

\newcommand{\bdXu}  {{\ensuremath{\dot  {\bXu}}}}
\newcommand{\bddXu} {{\ensuremath{\ddot {\bXu}}}}
\newcommand{\bdddXu}{{\ensuremath{\dddot{\bXu}}}}
\newcommand{\bdxu}  {{\ensuremath{\underline{\bxu}}}}

\newcommand{\bdYu}  {{\ensuremath{\dot  {\bYu}}}}
\newcommand{\bddYu} {{\ensuremath{\ddot {\bYu}}}}
\newcommand{\bdddYu}{{\ensuremath{\dddot{\bYu}}}}

\newcommand{\bdZu}  {{\ensuremath{\dot  {\bZu}}}}
\newcommand{\bddZu} {{\ensuremath{\ddot {\bZu}}}}
\newcommand{\bdddZu}{{\ensuremath{\dddot{\bZu}}}}

\newcommand{\bdH}   {{\ensuremath{\dot  {\bH}}}}
\newcommand{\bddH}  {{\ensuremath{\ddot {\bH}}}}
\newcommand{\bdddH} {{\ensuremath{\dddot{\bH}}}}

\newcommand{\bdh}   {{\ensuremath{\dot  {\bh}}}}
\newcommand{\bddh}  {{\ensuremath{\ddot {\bh}}}}
\newcommand{\bdddh} {{\ensuremath{\dddot{\bh}}}}

\newcommand{\bdR}   {{\ensuremath{\dot  {\bR}}}}
\newcommand{\bddR}  {{\ensuremath{\ddot {\bR}}}}
\newcommand{\bdddR} {{\ensuremath{\dddot{\bR}}}}

\newcommand{\bdD}   {{\ensuremath{\dot  {\bD}}}}
\newcommand{\bddD}  {{\ensuremath{\ddot {\bD}}}}
\newcommand{\bdddD} {{\ensuremath{\dddot{\bD}}}}

\newcommand{\bdB}   {{\ensuremath{\dot  {\bB}}}}
\newcommand{\bddB}  {{\ensuremath{\ddot {\bB}}}}
\newcommand{\bdddB} {{\ensuremath{\dddot{\bB}}}}

\newcommand{\bchi}  {\ensuremath{\textit{\textbf{x}}}}
\newcommand{\bChi}  {\ensuremath{\boldsymbol{\chi}}}

\newcounter{remarkCounter}
\stepcounter{remarkCounter}
\newcommand{\newRemark}[2]{\textit{ {\bf Remark \arabic{remarkCounter}}: (\textbf{#1}): #2 .} \stepcounter{remarkCounter}}

\newtheorem{definition}{Definition}
\newtheorem{assumption}{Assumption}
\newtheorem{proposition}{Proposition}

\date{}
\maketitle

\begin{abstract} Sensor calibration is an indispensable task in any networked cyberphysical system. In this paper, we consider a sensor network plagued with offset errors, measuring a rank-$1$ signal subspace, where each sensor collects measurements under a linear model with additive zero-mean Gaussian noise. Under varying assumptions on the underlying noise covariance, we investigate the effect of using an arbitrary reference for estimating the sensor offsets, in contrast to the `average of all the unknown offsets' as a reference. We first show that the \emph{average} reference yields an efficient minimum variance unbiased estimator. If the underlying noise is homoscedastic in nature, then we prove the \emph{average} reference yields a factor $2$ improvement on the variance, as compared to any arbitrarily chosen reference within the network. Furthermore, when the underlying noise is independent but not identical, we derive an expression for the improvement offered by the \emph{average} reference. We demonstrate our results using the problem of clock synchronization in sensor networks, and discuss directions for future work. 
\end{abstract}


\section{Introduction}
Sensors play a vital role in the burgeoning fields of internet of things (IoT) \cite{gubbi2013_iot}, networked cyberphysical systems \cite{gunes2014_cyberPhysicalSystems}, and Wireless Sensor Networks (WSN) \cite{sharma2013_wsn_issues}, with diverse applications e.g., environmental monitoring \cite{xie2017review}, multi-agent formation control \cite{oh2015survey,li2023}, remote sensing \cite{weiss2020_agriculture} and space systems \cite{rashvand2014_wsn_space,rajan2016space}, to name a few. Network-wide sensor calibration is a ubiquitous challenge in these applications, which is quintessential for accurately measuring the underlying signal of interest \cite{ling2018self,DEVITO2020_fieldCalibration}. Under the assumption that the ground truth lies in a known lower-dimensional subspace of the measurement subspace, the sensor gains \cite{balzano2007blind,dorffer2016blind,stankovic2018consensus} and the sensor drifts \cite{wang2016blind-drift,yang2021lightweight-drift,zaidan2022intelligent-drift} can be uniquely determined using blind calibration techniques. However, in the absence of a known reference, the offset estimation is typically infeasible, and thus leads to an ill-posed problem. In practise, in the absence of an external reference, one (or many) of the sensors within the network is chosen as a reference to make the system identifiable \cite{lipor2014robust}. More recently, in pursuit of an efficient minimum variance unbiased estimator, the average of the calibration parameters (e.g., offsets) is proposed as a reference in various applications, e.g., in antenna calibration \cite{wijnholdsConstrained06}, in clock synchronization \cite{rajan2015joint}, and in sensor calibration of air-quality networks \cite{rajan2018reference}. Although the \emph{average} reference yields optimality for various data models, it is unclear if it yields an optimal solution for a generalized linear model for offset estimation, and if so, how much is the improvement on performance as compared to any arbitrarily chosen reference within the network ? 



In this curiosity driven research, we consider a sensor network plagued with offset errors, where each sensor collects measurements under a linear Gaussian model, and measures an underlying signal subspace of rank-$1$. If the underlying  noise covariance is wide-sense stationary across the measurements, we show that the \emph{average} reference is the optimal reference for estimating the unknown sensor offsets in the network. Furthermore, when the  noise is \textit{i.i.d}, we show that the $\emph{average}$ reference offers a factor $2$ improvement in the variance of any unbiased estimator w.r.t. any arbitrarily chosen sensor reference in the network. Finally, when the  noise is independent across the sensors, but not identical, we derive an expression for the improvement offered by the \emph{average} reference in contrast to any other sensor within the network. 

\textit{Notation and properties:} The Kronecker product is indicated by $\otimes$, the transpose operator by ($\cdot)^T$ and $\equiv$ denotes equality by definition. $\b1_N$ and $\bzero_N$ denote a column vector of ones and zeros respectively. $\bI_N$ is an identity matrix, $\bdiag(\ba)$ represents a diagonal matrix with elements of vector $a$ along the diagonal, and $\text{vec}(.)$ indicates vectorization. then we have \begin{eqnarray}
    \text{vec}(\bA\bB\bC) &=& (\bC^T \otimes \bA)\text{vec}(\bB),
    \label{eq:vec_ABC} \\ 
    (\bA\otimes\bB)(\bC\otimes \bD)&=&(\bA\bC \otimes \bB\bD),
    \label{eq:prod_of_kron_prod} \\
    (\bA + \bb\bb^T)^{-1} 
    &=& \bA^{-1} - \mu^{-1}\bA^{-1}\bb\bb^T\bA^{-1},
    \label{eq:sherman-morrison}
\end{eqnarray} where $\bA,\bB,\bC$ and $\bD$ are of appropriate dimensions, and (\ref{eq:sherman-morrison}) holds for $\mu=(1 + \bb^T\bA^{-1}\bb) \ne 0$ and if $\bA$ is non-singular.





\subsection{Preliminaries}
Prior to modeling the sensor calibration problem, we briefly state the theoretical lower bound on the variance of any unbiased estimator under parametric constraints in Theorem~\ref{thm: ccrb}, and in Theorem~\ref{thm: optimal_constraint} we give the conditions for an optimal constraint set to yield an \emph{efficient} estimator.

\begin{theorem}[Constrained \Cramer\ Rao Bound (CCRB)] \label{thm: ccrb}Consider a consistent set of $k$ continuously differentiable constraints on $\btheta$ \ie $\bc(\btheta) = 0$, and let $\bC(\btheta)=\partial \bc(\btheta) / \partial \btheta^T$ be the gradient matrix which is full row rank, then the Constrained \Cramer\ Rao lower Bound (CCRB) on the variance of any unbiased estimator exists, and is bounded by \begin{equation}
    \mathbb{E} \big\{ (\hat{\btheta}-\btheta))(\hat{\btheta}-\btheta))^T \big\}
    \equiv \bSigma_{\hat{\theta}}
    \ge \bU\big(\bU^T\bF\bU\big)^{-1}\bU^T, 
    \label{eq:CCRB}
\end{equation} where $\bU$ spans the null space of the gradient matrix $\bC(\btheta)$, and the Fisher information matrix (FIM) is given by $\bF = -\mathbb{E} \{ \partial^2 \log p(\by;\btheta)/\partial \btheta^2 \}$, where $p(.)$ is the {p.d.f.} of the measurements $\by$, which meets certain regularity conditions \cite{Kay1993}.
\end{theorem}
\begin{proof}
\cite[Theorem 1.]{stoica1998}
\end{proof}

\begin{theorem}[Optimal constraint set] \label{thm: optimal_constraint} Any feasible set of linearly independent vectors which span the nullspace of the FIM forms an optimal constraint set.
\end{theorem}
\begin{proof} Let the spectral decomposition of the FIM be $\bF = \bV\bLambda\bV^T$ where $\bLambda$ is a diagonal matrix containing the non-zero eigenvalues, and let $\bV$ contain the corresponding eigenvectors. Now, consider a constraint matrix $\bar{\bC}$, such that the orthogonal basis for the null space of $\bar{\bC}$ spans the range of $\bV$, then substituting for $\bU$ with $\bV$ in (\ref{eq:CCRB}),we have \begin{equation}
\trace\big(\bSigma_{\theta}\big)
= \trace\big[\bV(\bV^T\bF\bV)^{-1}\bV^T \big] \nonumber \\
\overset{(a)}{=} \trace\big[\bLambda^{-1}\big] \equiv \trace\big(\bF^\dagger \big),
\end{equation} where we substitute for $\bF= \bV\bLambda\bV^T$ to obtain (a), exploit the cyclic nature of the trace operator, and the orthonormal property of $\bV$. Observe that the unconstrained FIM yields the lowest achievable variance for any unbiased estimator, and thus $\bar{\bC}$ is an optimal constraint set \cite{rajan2015joint}.
\end{proof}





\section{Data Model}
Consider a network of $N$ sensor nodes, where each sensor collects $K \ge N$ measurements. Let the \emph{true} signal impinging on the $N$ sensors be given by a $K \times N$ matrix $\bS = [\bs_1, \bs_2, \hdots, \bs_N] $, where the ground-truth $\bS$ lies in a lower dimensional subspace, say $r$, such that $r < N$. Furthermore, let $\bGamma$ be a known $(N-r) \times N$ projection matrix of rank $r$, which spans the orthogonal complement of the signal subspace, such that $\bS\bGamma^T= \bzero$, then using (\ref{eq:vec_ABC}), we have \begin{equation}
     (\bGamma \otimes \bI_K)vec(\bS)= \bar{\bGamma}\bs=\bzero, \label{eq:known_subspace}
\end{equation} which is the underlying premise for blind calibration in sensor networks \cite{ling2018self, balzano2007blind,dorffer2016blind,stankovic2018consensus}. In practise, if $\bGamma$ is unknown, then the underlying subspace can be learnt over time \cite{yan2017learning}. We now consider a scenario where the $N$ sensors are plagued with offset errors, and the measurements are corrupted with noise. The measurements of the $n$th sensor ($1 \le n \le N$) are denoted by  $\by_n = \bs_n + \theta_n\b1_K + \bEta_n$, where $\bs_n$ is the \emph{true} signal of length $K$, $\theta_n$ is the unknown sensor offset and $\bEta_n$ is the stochastic noise on the measurements. Extending for all $N$ sensors we have the following data model \begin{equation}
\label{eq:y_model}
\by= \bs + \bH\btheta + \bEta,
\end{equation} where $\by = [\by^T_1,\by^T_2,\hdots,\by^T_N]^T$, $\bs = [\bs^T_1,\bs^T_2,\hdots,\bs^T_N]^T$, $ \btheta = [\theta_1, \theta_2, \hdots, \theta_N]^T$, $\bH \equiv \bI_N \otimes \b1_K$, and $\bEta = [\bEta^T_1,\bEta^T_2,\hdots,\bEta^T_N]^T$. Furthermore, let the noise on the system (\ref{eq:y_model}) be zero-mean Gaussian, where the covariance function is homogeneous \ie wide sense stationary across the $K$ measurements \ie $\bEta \sim \cN(\bzero, \bSigma \otimes \bI_K)$. Then exploiting the assumption (\ref{eq:known_subspace}), multiplying (\ref{eq:y_model}) by the known $\bar{\bGamma}$, and rearranging the terms we have \begin{equation}
    \bar{\bGamma}\by  
    \sim \cN\big(\bar{\bGamma}\bH\btheta,
    \bar{\bGamma}(\bSigma \otimes \bI_K )\bar{\bGamma}^T \big).
    \label{eq:datamodel}
\end{equation}The Fisher information (FIM) for the above data model is straightforward \cite{Kay1993}, and is given by \begin{equation}
\bF 
= \bH^T\bar{\bGamma}^T(\bGamma\bSigma\bGamma^T \otimes \bI_{K})^{-1}\bar{\bGamma}\bH 
= K\bGamma^T(\bGamma\bSigma\bGamma^T)^{-1}\bGamma
    \label{eq:FIM_blindcal}
\end{equation} where we substitute for $\bar{\bGamma}= \bGamma \otimes \bI_K$, $\bH= \bI_N \otimes \b1_K$ and use the relation (\ref{eq:prod_of_kron_prod}). Since $\bGamma$ is rank deficient, note that the FIM is also rank deficient. Let $\bar{\bC}$ span the nullspace of the FIM, then from Theorem~\ref{thm: optimal_constraint}, the optimal parametric constraint set is given by $\bar{\bC}^T\btheta = \bd$, where $\bd$ is a known response vector of length $N-r$. If sufficient data is collected \ie $K\ge N$, and if the covariance $\bSigma$ is known, then $\bar{\bC}$ can be designed, which would lead to a \emph{data-driven} reference.
\subsection{Single source}
In the following section, we look at a special case of (\ref{eq:datamodel}), where the sensors measure an identical signal \eg densely deployed sensor array measuring air quality \cite{spinelle2017field}. In these scenarios, the signal subspace for the $k$th measurement  ($\forall\ k \le K$) across all $N$ nodes, spans a column vector of ones i.e., $\b1_N$. Subsequently, the projection matrix takes the form $\bGamma\equiv \bU^T_2= \begin{bmatrix}
    -\b1_{N-1} & \bI_{N-1}
\end{bmatrix}^T$, and (\ref{eq:datamodel}) and (\ref{eq:FIM_blindcal}) simplify to \begin{eqnarray}
    \bz &\equiv& \bar{\bU}^T_2\by \sim\ 
    \cN\big(\bar{\bU}^T_2\bH\btheta,\bU^T_2\bSigma\bU_2 \otimes \bI_K\big),
    \label{eq:single_source_datamodel} \\
    \bF &=& K\bU_2(\bU^T_2\bSigma\bU_2)^{-1}\bU^T_2.
    \label{eq:FIM_rank1}
\end{eqnarray} where $\bar{\bU}_2 = \bU_2 \otimes \bI_K$. Observe that since $\bU_2$ is rank deficient, $\bF$ is also rank deficient by least $r=1$, which is intuitively expected, since at least $1$ reference is needed to uniquely estimate all the sensor offsets. 


\subsection{Constraints} A conventional solution to estimate the offset is to arbitrarily assume one of the sensors as a reference, which without a loss of generality, we choose as sensor $1$. The corresponding parametric constraint, the gradient vector and the orthonormal bases for the nullspace of the gradient vector are given as \begin{equation}
    \theta_1= 0, \quad 
    \bc_1 = [1,\bzero_{N-1}^T]^T, \quad
    \bU_1 = \begin{bmatrix}
    \bzero^T_{N-1} \\
    \bI_{N-1}
    \end{bmatrix},\quad 
    \label{eq:constraints_1}
\end{equation} respectively, where we use the subscript $1$ to denote a \emph{single} reference. \ Alternatively, the average of all unknown offsets could be used as a reference, which has been proposed in various applications \cite{wijnholdsConstrained06,rajan2015joint,rajan2018reference}. The constraint, the corresponding gradient and the bases for the nullspace of the gradient are then given as \begin{equation}
    \frac{1}{N}\sum^{N}_{n=1}\theta_n = 0, \quad 
    \bc_2 = N^{-1}\b1_N, \quad
    \bU_2 = \begin{bmatrix}
    -\b1^T_{N-1} \\
    \bI_{N-1}
    \end{bmatrix},\quad 
    \label{eq:constraints_2}
\end{equation} respectively, where the subscript $2$ is used to denote the \emph{average} reference. Increasing the number of constraints would further improve the performance of any estimator, however, we limit our discussion in this paper to a single constraint, which is sufficient for identifiabilty.

\subsection{Effect on sensor bias}\label{section:bias} If the network is calibrated, then observe that $\theta_n=0\ \forall\ n\le N$. Now, by choosing a reference \eg sensor $1$, we force the corresponding offset to $0$ \ie $\theta_1 = 0$, and hence implicitly introduce a bias $\phi_1$ into our final estimator, which reflects the unidentifiable offset of sensor $1$. Along similar lines, (\ref{eq:constraints_2}) introduces a bias, which is the mean of  $\phi_n \ \forall\ n\le N$, where $\phi_n$ is the implicit bias of the $n$th sensor. Now, if the mean of the underlying offsets are centered around $0$, then as $N \rightarrow \infty$, observe that (\ref{eq:constraints_2}) minimizes the overall network bias. On the contrary, by using the single reference (\ref{eq:constraints_1}), we rely entirely on the performance of sensor $1$, which leads to a single point of failure. In the next section we study the effect of these constraints on the variance of an unbiased estimator for $\btheta$.


\section{Optimal reference} We now aim to compare the performance of any constrained unbiased estimator for (\ref{eq:single_source_datamodel}), under varying assumptions of the covariance $\bSigma$. More concretely, let $\bSigma_1$ and $\bSigma_2$ denote the lower bound on the variance of any unbiased estimator under the constraints (\ref{eq:constraints_1}) and (\ref{eq:constraints_2}) respectively, then using (\ref{eq:CCRB}), we aim to evaluate and analyze 
\begin{subequations}
\begin{align}
\trace(\bSigma_1)
=& \trace\big[(\bU^T_1\bF \bU_1)^{-1} \big] \label{eq:trace_ccrb_1}, \\
\trace(\bSigma_2) 
=&\trace\big[(\bU^T_2\bF \bU_2)^{-1}\bU^T_2\bU_2\big] \label{eq:trace_ccrb_2},
\end{align} \label{eq:trace_ccrb}\end{subequations}  where we exploit the cyclic nature of the trace operator, and the property $\bU_1^T\bU_1=\bI$ in (\ref{eq:trace_ccrb_1}). To this end, we have the following result.

\begin{theorem}[Optimal reference] \label{thm:optimal_ref} Consider an network of $N$ sensors, where each sensor collects $K$ measurements based on the data model (\ref{eq:single_source_datamodel}), then the following statements hold.
\begin{itemize}
    \item[(a)] Heteroscedasticity: The optimal reference for estimating the unknown sensor offsets is the \emph{average} reference (\ref{eq:constraints_2})
    \item[(b)] Homoscedasticity:  If $\bSigma= \sigma^2\bI$, the optimal reference for estimating the unknown offsets, outperforms any arbitrarily chosen reference in the network, by a factor $2$.
    \item[(c)] Independent, but not identical: In the special case of \begin{equation}
        \bSigma= \bdiag(\sigma^2_1, \sigma^2_2, \hdots, \sigma^2_N),
        \label{eq:diagonal_cov_matrix}
    \end{equation} the variance of the optimal unbiased estimator outperforms any arbitrarily chosen sensor reference by \begin{equation} \dfrac{N}{K}\dfrac{\sigma_i^2 \sum^N_{n=1,n\ne i}\sigma^2_n}{\sigma_i^2 + \sum^N_{n=1,n\ne i}\sigma^2_n},
\end{equation} where $\sigma_i^2$ is the variance of the chosen sensor reference.
\end{itemize} 
\end{theorem} 

    

\begin{proof}(Proof of Theorem~\ref{thm:optimal_ref}(a)) Observe from (\ref{eq:FIM_rank1}) and (\ref{eq:constraints_2}) that $\bF\bc_2= \bzero$, since $\b1_N$ spans the nullspace of $\bU^T_2$, and hence using Theorem~\ref{thm: optimal_constraint}, the \emph{average} (\ref{eq:constraints_2}) is the optimal reference. In other words, (\ref{eq:constraints_2}) yields a minimum variance unbiased estimator (MVUE) for the data model (\ref{eq:single_source_datamodel}).
\end{proof}

In the following sections, we present the proofs of Theorem~\ref{thm:optimal_ref}(b) and Theorem~\ref{thm:optimal_ref}(c).

\subsection{Homoscedasticity}
\begin{proof}(Proof of Theorem~\ref{thm:optimal_ref}(b)) The optimal reference for estimating the sensor offsets is the \emph{average} of all the unknown offsets as proved in Theorem~\ref{thm:optimal_ref}(a). Hence, when $\bSigma= \sigma^2\bI$, it suffices to show \begin{equation}
    \delta\equiv 
    \trace(\bSigma_2)/\trace(\bSigma_1)= 0.5,
    \label{eq:delta_homo_factor}
\end{equation} for a homoscedastic system, where $\trace(\bSigma_1)$ and $\trace(\bSigma_2)$ are given in (\ref{eq:trace_ccrb}). In this scenario, the FIM in (\ref{eq:FIM_rank1}) simplifies to \begin{equation}
    \bF
    = \sigma^{-2} K \bU_2(\bU^T_2\bU_2)^{-1}\bU^T_2
    = \sigma^{-2} K \bU_2\bPsi^{-1}\bU^T_2,
    \label{eq:FIM_rank1_homo}
\end{equation} where we define \begin{subequations}
\begin{align}
    \bPsi \equiv\ & \bU^T_2\bU_2 = \bI_{N-1} + \b1_{N-1}\b1^T_{N-1},
    \label{eq:U2U2} \\
    \bPsi^{-1} \equiv\ & (\bU^T_2\bU_2)^{-1}= \bI_{N-1} - N^{-1}\b1_{N-1}\b1^T_{N-1}, 
    \label{eq:U2U2_inv}
\end{align}
\label{eq:PSI}
\end{subequations} and use the Sherman-Morrison formula (\ref{eq:sherman-morrison}) to obtain (\ref{eq:U2U2_inv}). Now, substituting the FIM (\ref{eq:FIM_rank1_homo}) and $\bU_1$  (\ref{eq:constraints_1}) in (\ref{eq:trace_ccrb_1}), the CCRB for a single reference is  \begin{align}
\trace(\bSigma_1) 
=&\ \frac{\sigma^2}{K}\trace\bigg[(\bU^T_1\bU_2\bPsi^{-1}\bU^T_2\bU_1)^{-1}\bigg] 
\overset{(a)}{=} \frac{\sigma^2}{K}\trace(\bPsi), \nonumber \\
\overset{(b)}{=}&\ \frac{\sigma^2}{K}\trace[\bI_{N-1} + \b1_{N-1}\b1^T_{N-1}] 
= \frac{2\sigma^2}{K}(N-1). 
\label{eq:trace_homo_1}
\end{align} where we use the property $\bU^T_1\bU_2 = \bI$ in (a), and substitute for $\bPsi$ (\ref{eq:U2U2}) to obtain (b) in (\ref{eq:trace_homo_1})
). Along similar lines, the CCRB for the average reference (\ref{eq:constraints_2}) is obtained by substituting for the FIM (\ref{eq:FIM_rank1_homo}) in (\ref{eq:trace_ccrb_2}), which yields \begin{align}
\trace(\bSigma_2) 
=&\ \frac{\sigma^2}{K}\trace\big[(\bU^T_2(\bU_2\bPsi^{-1}\bU^T_2)\bU_2)^{-1} \bU^T_2\bU_2\big] \nonumber, \\
=&\ \frac{\sigma^2}{K}\trace(\bPsi)= \frac{\sigma^2}{K}(N-1),
\label{eq:trace_homo_2}
\end{align} where we use the definitions (\ref{eq:PSI}). Finally, from (\ref{eq:trace_homo_1}) and (\ref{eq:trace_homo_2}), we have (\ref{eq:delta_homo_factor}), and hence proved.
\end{proof}

\subsection{Independent, but not identical} We now consider the scenario  (\ref{eq:diagonal_cov_matrix}), where without loss of generality, we assume $\sigma^2_1 \le \sigma^2_2 \le \hdots \le \sigma^2_N$. Following immediately, sensor $1$ (with the lowest variance) is an appropriate reference for calibration. To this end, the constraints (\ref{eq:constraints_1}) hold, and thus proving Theorem~\ref{thm:optimal_ref}(c) is equivalent to showing \begin{equation}
    \trace(\bSigma_1) - \trace(\bSigma_2) 
    = \dfrac{N}{K}\dfrac{\sigma_1^2 \sum^N_{n=2}\sigma^2_n}{\sigma_1^2 + \sum^N_{n=2}\sigma^2_n}.
    \label{eq:hetero_factor}
\end{equation}

\begin{figure*}[t]
    \centering
    \subfigure[]{\includegraphics[width=50mm]{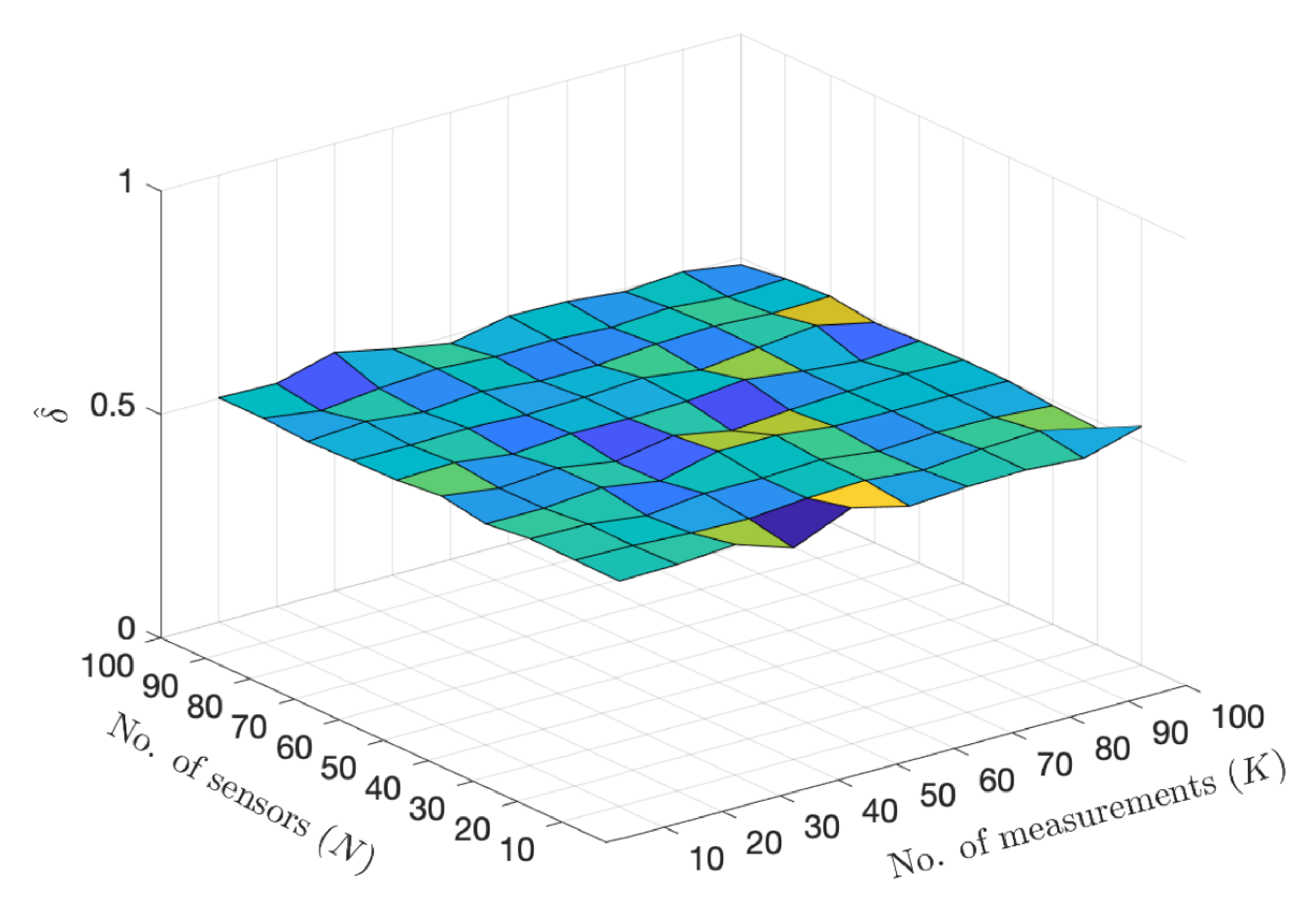}}
    \hspace{3mm}
    \subfigure[]{\includegraphics[width=40mm]{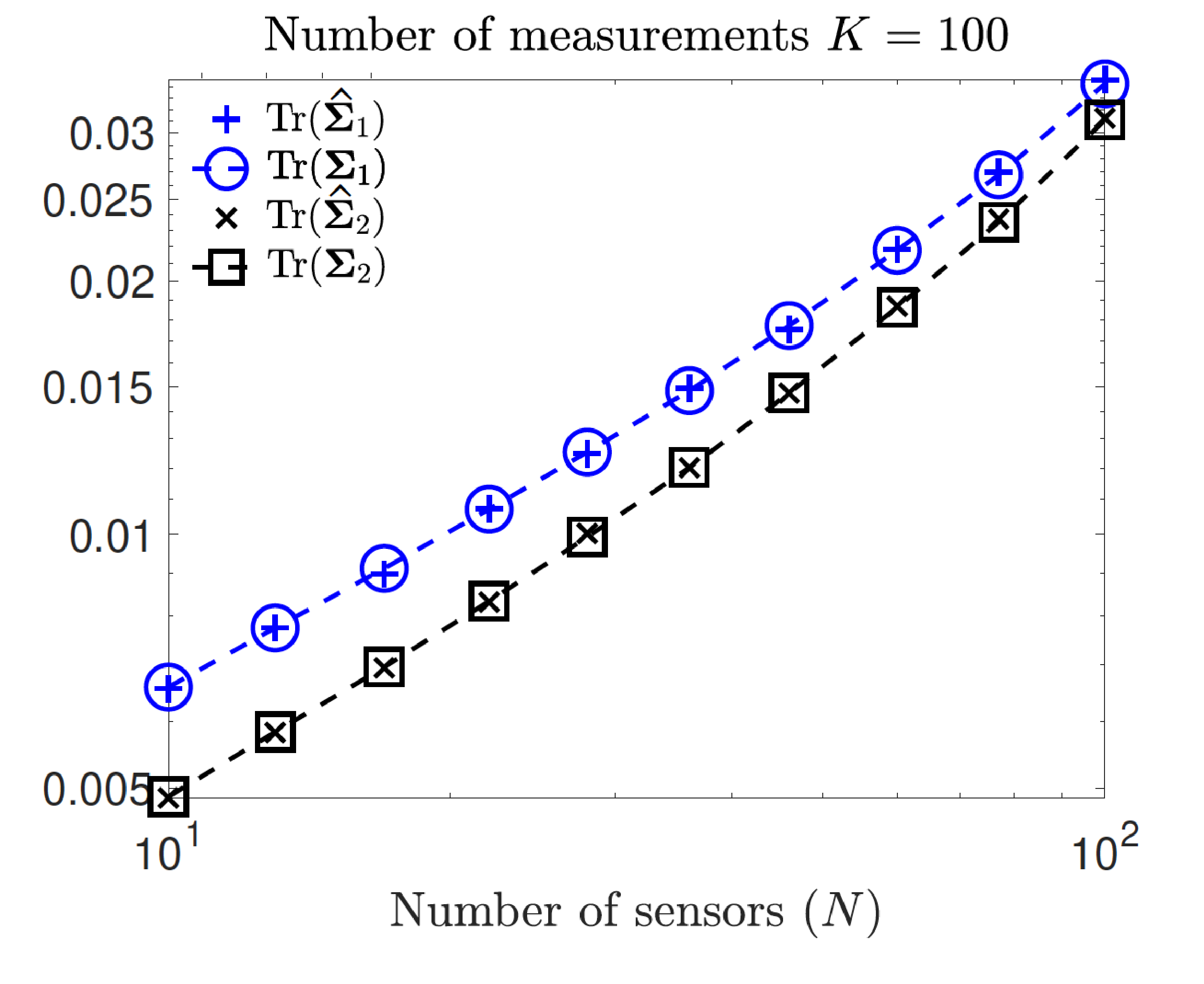}}
    \hspace{3mm}
    \subfigure[]{\includegraphics[width=40mm]{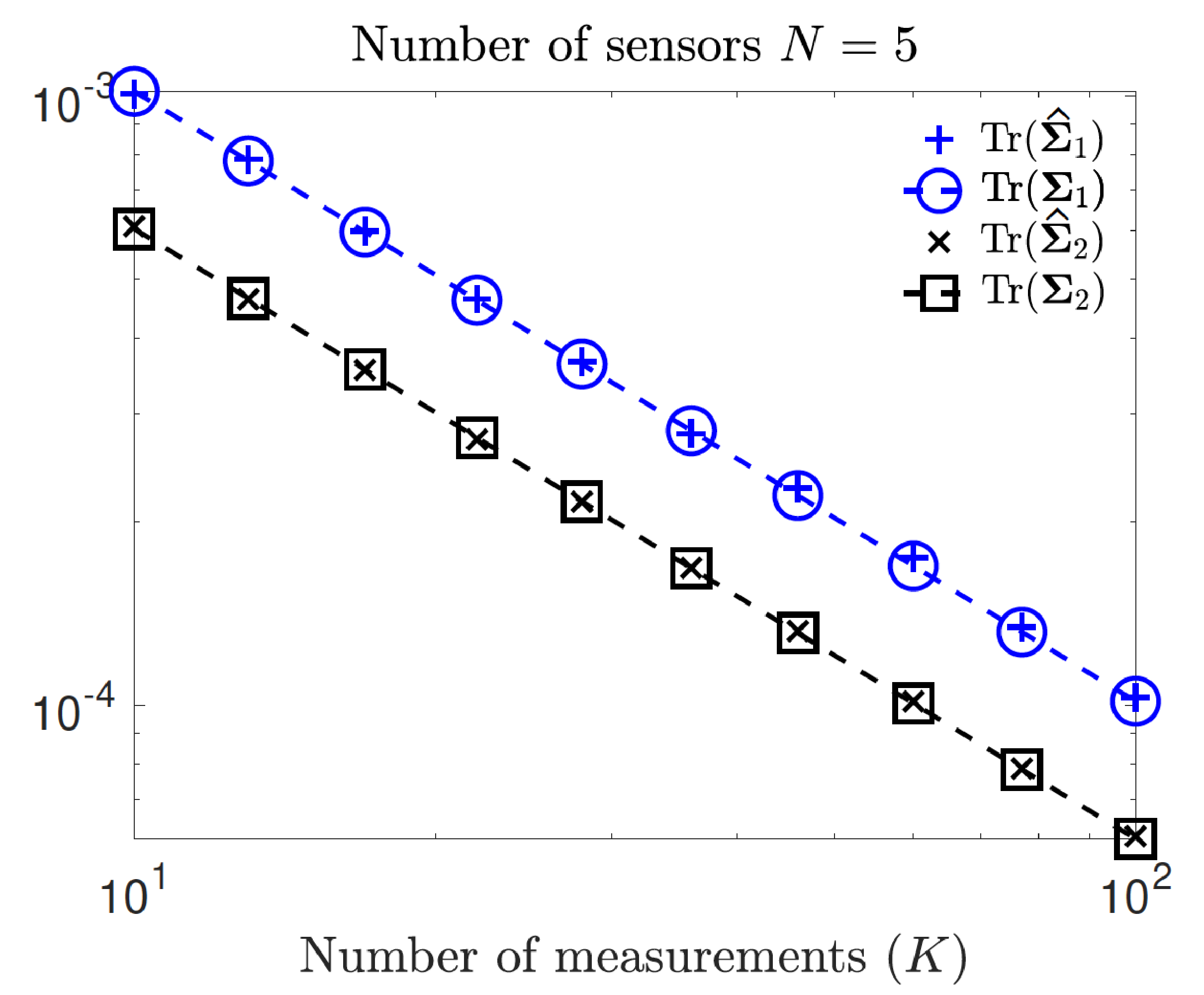}}
    \caption{\small Figure (a) shows an estimate of (\ref{eq:delta_homo_factor}) \ie $\hat{\delta}= \hat{\bSigma}_2/\hat{\bSigma}_1$, for varying $K$ and $N$, illustrating the factor $2$ improvement using the \emph{average} reference in contrast to a single reference in the network, under \emph{i.i.d} Gaussian noise. Under the assumption of independent but not identical Gaussian noise, (b) and (c) show that the estimator performance, for varying $K$ (fixed $N$) and varying $N$ (fixed $K$) respectively.}
    \label{fig:all}
\end{figure*}

\begin{proof}(Proof of Theorem~\ref{thm:optimal_ref}(c)) The CCRB for the single reference scenario is obtained by substituting for the FIM (\ref{eq:FIM_rank1}) and $\bU_1$ (\ref{eq:constraints_1}) in (\ref{eq:trace_ccrb_1}), which yields  \begin{align}
    \trace(\bSigma_1) 
    =& K^{-1} \trace\bigg[\big(\bU^T_1\bU_2(\bU^T_2\bSigma\bU_2)^{-1}\bU^T_2\bU_1\big)^{-1}\bigg] \nonumber \\
    =& K^{-1} \trace(\bU^T_2\bSigma\bU_2\big) = K^{-1} \trace(\bOmega),
    \label{eq:trace_hetro_1}
\end{align} where we exploit $\bU^T_1\bU_2 = \bI$ and introduce \begin{equation}
    \bOmega 
    \equiv \bU^T_2\bSigma\bU_2
    = \bar{\bSigma} + \sigma^2_1\b1_{N-1}\b1_{N-1}^T.
    \label{eq:omega}
\end{equation} The CCRB for the average reference (\ref{eq:constraints_2}), using (\ref{eq:trace_ccrb_2}) is, \begin{align}
\trace(\bSigma_2)
=& \trace\bigg[\big(\bU^T_2(\bU_2(\bU^T_2\bSigma\bU_2)^{-1}\bU^T_2\big)\bU_2)^{-1}\bU^T_2\bU_2 \bigg] \nonumber \\
=& K^{-1}\trace\big[(\bPsi\bOmega^{-1}\bPsi)^{-1}\bPsi\big] =\ K^{-1} \trace\big(\bOmega\bPsi^{-1}\big),
\label{eq:trace_hetro_2}
\end{align} where we use the definitions (\ref{eq:PSI}) and (\ref{eq:omega}). Subtracting (\ref{eq:trace_hetro_2}) from (\ref{eq:trace_hetro_1}), we have $\trace(\bSigma_1) - \trace(\bSigma_2)= NK^{-1}\trace\big(\bOmega\b1_{N-1}\b1^T_{N-1}\big)$. Substituting for $\bOmega$ from (\ref{eq:omega}),  and further manipulations yield (\ref{eq:hetero_factor}), and hence proved. 
\end{proof}

\section{Simulations : Clock synchronization} We consider a sensor network of $N$ clocks, where each clock is plagued with an offset, and collects $K$ measurements \cite{freris10,rajanCAMSAP11}. Let $\bt=[t_1, t_2, \hdots, t_K]$ denote the \emph{true} time over $K$ instances, and let $\bT_i= [T_{i,1},T_{i,2},\hdots, T_{i,K}]^T$ be the $K$ measurements at the $i$th clock, then we have $\bT_{i} = \bt + \theta_i\b1_K + \bEta_i $, where $\theta_i$ and  $\bEta_i$ are the clock offset and the stochastic noise plaguing the $i$th clock respectively. Let $\bT=[\bT_1,\bT_2,\hdots,\bT_N]^T$ populate all the time measurements from the $N$ sensors, then observe that the data model (\ref{eq:single_source_datamodel}) holds where $\by = vec{(\bT)}$. Subsequently, to estimate the unknown sensor offsets, we aim to minimize the constrained weighted least squares formulation \begin{equation}
    \hat{\btheta} 
    = \argmin_{\btheta}\norm{\bW\bar{\bP}(\by - \bH\btheta)}^2_2
    \quad
    \text{s.t.}\ \bc(\btheta),
    \label{eq:cls}
\end{equation} where $\bW$ is a weighting matrix, $\bar{\bP}= \bP \otimes \bI_K$, where $\bP$ is a $N \times N$ centering matrix, $\bc$ is the gradient of the parametric constraint. The cost function has a closed form solution for a non-empty constraint set $\bc$, such that $\begin{bmatrix} \bH^T\bar{\bP}^T & \bc^T\end{bmatrix}^T$ is full rank \cite{boydConvexOptimization}. The estimator (\ref{eq:cls}) under constraints (\ref{eq:constraints_1}) and (\ref{eq:constraints_2}) are denoted by $\hat{\btheta}_1$ and $\hat{\btheta}_2$ respectively. Observe that, with $\bW = \bSigma^{-1/2}$, (\ref{eq:cls}) yields the optimal estimator \ie the minimum variance unbiased estimator for (\ref{eq:single_source_datamodel}), independently for the respective constraints (\ref{eq:constraints_1}) and (\ref{eq:constraints_2}).

We perform experiments to validate the derived lower bounds, using the performance of the optimal estimators (\ref{eq:cls}) under parametric constraints (\ref{eq:constraints_1}) and (\ref{eq:constraints_2}), for varying assumptions on the noise covariance matrix. In case of $\bSigma=\sigma^2\bI$, where $\bsigma^2= 10^{-3}$, we perform a Monte carlo experiment by linearly varying both the number of sensors and measurements from $10$ to $100$. We evaluate an estimate of (\ref{eq:delta_homo_factor}) \ie $\hat{\delta}= \hat{\bSigma}_2/\hat{\bSigma}_1$ for each parameter set, where $\hat{\bSigma}_1$ and $\hat{\bSigma}_2$ are the estimated covariances under constraints (\ref{eq:constraints_1}) and (\ref{eq:constraints_2}) respectively, over $N_{exp}=500$ Monte Carlo runs. Figure~\ref{fig:all}(a) visually illustrates the factor $2$ improvement, which is invariant over the number of measurements or sensors, as stated in Theorem~\ref{thm:optimal_ref}(b).


When the underlying noise is independent, but not identical (\ref{eq:diagonal_cov_matrix}), we consider two independent simulations \ie a fixed number of sensors $N=5$, and increasing number of measurements from $K=10$ to $100$ in Figure~\ref{fig:all}(b), and a fixed number of measurements $K=100$, and varying the number of sensors from $N=10$ to $100$ in Figure~\ref{fig:all}(c). Note that the estimated variances meet the CCRBs, as expected. Furthermore, observe from (\ref{eq:trace_hetro_1}) and (\ref{eq:trace_hetro_2}), that for a fixed $N$, the CCRB linearly decreases with increasing $K$. However, for varying $N$ with a fixed $K$, the distribution of variances impacts the performance of the estimator. If the probability distribution of the variances is available, then further deductions can be made, using Theorem~\ref{thm:optimal_ref}(c).

Note that in the absence of a known reference, our data model in the presence of constraints (\ref{eq:constraints_1}) and (\ref{eq:constraints_2}), will implicitly introduce a bias which depends on the offsets as discussed in Section~\ref{section:bias}. The average reference will minimize this bias for larger $N$, or if the underlying sensor offset distribution is symmetrical e.g., standard normal distribution. In more practical scenarios, despite smaller number of sensors e.g., $N=5$, the average reference has shown to outperform a single reference for limited datasets \cite{rajan2018reference}.

\section{Conclusion} In this paper, we investigated the effect of the choice of reference in offset estimation of a sensor network, where sensor measurements are plagued with zero-mean Gaussian noise. We show that the average of all the unknown offsets as a parametric constraint \ie the \emph{average} reference, yields the MVUE. In the particular case of homoscedasticity, the improvement on the variance of the MVUE is by a factor $2$, as compared to any single reference, which is not readily intuitive at the outset of the problem formulation. In our future work, we aim to investigate an optimal constraint set for other (non-linear) data models and probability distributions, including a Bayesian inference framework when the signal subspace is unknown, and where the choice of the prior on the unknown parameters would play a key role in the optimality of the constrained estimator. \\




\newpage
\bibliography{ref.bib}

\begin{thebibliography}{10}
\providecommand{\url}[1]{#1}
\csname url@samestyle\endcsname
\providecommand{\newblock}{\relax}
\providecommand{\bibinfo}[2]{#2}
\providecommand{\BIBentrySTDinterwordspacing}{\spaceskip=0pt\relax}
\providecommand{\BIBentryALTinterwordstretchfactor}{4}
\providecommand{\BIBentryALTinterwordspacing}{\spaceskip=\fontdimen2\font plus
\BIBentryALTinterwordstretchfactor\fontdimen3\font minus
  \fontdimen4\font\relax}
\providecommand{\BIBforeignlanguage}[2]{{%
\expandafter\ifx\csname l@#1\endcsname\relax
\typeout{** WARNING: IEEEtran.bst: No hyphenation pattern has been}%
\typeout{** loaded for the language `#1'. Using the pattern for}%
\typeout{** the default language instead.}%
\else
\language=\csname l@#1\endcsname
\fi
#2}}
\providecommand{\BIBdecl}{\relax}
\BIBdecl

\bibitem{gubbi2013_iot}
J.~Gubbi, R.~Buyya, S.~Marusic, and M.~Palaniswami, ``Internet of things (iot):
  A vision, architectural elements, and future directions,'' \emph{Future
  generation computer systems}, vol.~29, no.~7, pp. 1645--1660, 2013.

\bibitem{gunes2014_cyberPhysicalSystems}
V.~Gunes, S.~Peter, T.~Givargis, and F.~Vahid, ``A survey on concepts,
  applications, and challenges in cyber-physical systems,'' \emph{KSII
  Transactions on Internet and Information Systems (TIIS)}, vol.~8, no.~12, pp.
  4242--4268, 2014.

\bibitem{sharma2013_wsn_issues}
S.~Sharma, R.~K. Bansal, and S.~Bansal, ``Issues and challenges in wireless
  sensor networks,'' in \emph{2013 international conference on machine
  intelligence and research advancement}.\hskip 1em plus 0.5em minus
  0.4em\relax IEEE, 2013, pp. 58--62.

\bibitem{xie2017review}
X.~Xie, I.~Semanjski, S.~Gautama, E.~Tsiligianni, N.~Deligiannis, R.~T. Rajan,
  F.~Pasveer, and W.~Philips, ``A review of urban air pollution monitoring and
  exposure assessment methods,'' \emph{ISPRS International Journal of
  Geo-Information}, vol.~6, no.~12, p. 389, 2017.

\bibitem{oh2015survey}
K.-K. Oh, M.-C. Park, and H.-S. Ahn, ``A survey of multi-agent formation
  control,'' \emph{Automatica}, vol.~53, pp. 424--440, 2015.

\bibitem{li2023}
Z.~Li and R.~T. Rajan, ``{Geometry-Aware Distributed Kalman Filtering for
  Affine Formation Control under Observation Losses},'' in \emph{2023 26th
  International Conference on Information Fusion (FUSION)}, 2023, pp. 1--7.

\bibitem{weiss2020_agriculture}
M.~Weiss, F.~Jacob, and G.~Duveiller, ``Remote sensing for agricultural
  applications: A meta-review,'' \emph{Remote Sensing of Environment}, vol.
  236, p. 111402, 2020.

\bibitem{rashvand2014_wsn_space}
H.~F. Rashvand, A.~Abedi, J.~M. Alcaraz-Calero, P.~D. Mitchell, and S.~C.
  Mukhopadhyay, ``Wireless sensor systems for space and extreme environments: A
  review,'' \emph{IEEE Sensors Journal}, vol.~14, no.~11, pp. 3955--3970, 2014.

\bibitem{rajan2016space}
R.~T. Rajan, A.-J. Boonstra, M.~Bentum, M.~Klein-Wolt, F.~Belien, M.~Arts,
  N.~Saks, and A.-J. van~der Veen, ``Space-based aperture array for ultra-long
  wavelength radio astronomy,'' \emph{Experimental Astronomy}, vol.~41, pp.
  271--306, 2016.

\bibitem{ling2018self}
S.~Ling and T.~Strohmer, ``Self-calibration and bilinear inverse problems via
  linear least squares,'' \emph{SIAM Journal on Imaging Sciences}, vol.~11,
  no.~1, pp. 252--292, 2018.

\bibitem{DEVITO2020_fieldCalibration}
S.~{De Vito}, E.~Esposito, N.~Castell, P.~Schneider, and A.~Bartonova, ``On the
  robustness of field calibration for smart air quality monitors,''
  \emph{Sensors and Actuators B: Chemical}, vol. 310, p. 127869, 2020.

\bibitem{balzano2007blind}
L.~Balzano and R.~Nowak, ``Blind calibration of sensor networks,'' in
  \emph{Proceedings of the 6th international conference on Information
  processing in sensor networks}, 2007, pp. 79--88.

\bibitem{dorffer2016blind}
C.~Dorffer, M.~Puigt, G.~Delmaire, and G.~Roussel, ``Blind mobile sensor
  calibration using an informed nonnegative matrix factorization with a relaxed
  rendezvous model,'' in \emph{2016 IEEE International Conference on Acoustics,
  Speech and Signal Processing (ICASSP)}.\hskip 1em plus 0.5em minus
  0.4em\relax IEEE, 2016, pp. 2941--2945.

\bibitem{stankovic2018consensus}
M.~S. Stankovi{\'c}, S.~S. Stankovi{\'c}, K.~H. Johansson, M.~Beko, and L.~M.
  Camarinha-Matos, ``On consensus-based distributed blind calibration of sensor
  networks,'' \emph{Sensors}, vol.~18, no.~11, p. 4027, 2018.

\bibitem{wang2016blind-drift}
Y.~Wang, A.~Yang, Z.~Li, X.~Chen, P.~Wang, and H.~Yang, ``Blind drift
  calibration of sensor networks using sparse bayesian learning,'' \emph{IEEE
  Sensors Journal}, vol.~16, no.~16, pp. 6249--6260, 2016.

\bibitem{yang2021lightweight-drift}
L.~Yang and A.~Shami, ``A lightweight concept drift detection and adaptation
  framework for iot data streams,'' \emph{IEEE Internet of Things Magazine},
  vol.~4, no.~2, pp. 96--101, 2021.

\bibitem{zaidan2022intelligent-drift}
M.~A. Zaidan, N.~H. Motlagh, P.~L. Fung, A.~S. Khalaf, Y.~Matsumi, A.~Ding,
  S.~Tarkoma, T.~Pet{\"a}j{\"a}, M.~Kulmala, and T.~Hussein, ``Intelligent air
  pollution sensors calibration for extreme events and drifts monitoring,''
  \emph{IEEE Transactions on Industrial Informatics}, vol.~19, no.~2, pp.
  1366--1379, 2022.

\bibitem{lipor2014robust}
J.~Lipor and L.~Balzano, ``Robust blind calibration via total least squares,''
  in \emph{2014 IEEE International Conference on Acoustics, Speech and Signal
  Processing (ICASSP)}.\hskip 1em plus 0.5em minus 0.4em\relax IEEE, 2014, pp.
  4244--4248.

\bibitem{wijnholdsConstrained06}
S.~Wijnholds and A.-J. {van der Veen}, ``Effects of parametric constraints on
  the {CRLB} in gain and phase estimation problems,'' \emph{Signal Processing
  Letters, IEEE}, vol.~13, no.~10, pp. 620 --623, oct. 2006.

\bibitem{rajan2015joint}
R.~T. Rajan and A.-J. van~der Veen, ``Joint ranging and synchronization for an
  anchorless network of mobile nodes,'' \emph{IEEE Transactions on Signal
  Processing}, vol.~63, no.~8, pp. 1925--1940, 2015.

\bibitem{rajan2018reference}
R.~T. Rajan, A.~Das, J.~Romme, F.~Pasveer \emph{et~al.}, ``Reference-free
  calibration in sensor networks,'' \emph{IEEE sensors letters}, vol.~2, no.~3,
  pp. 1--4, 2018.

\bibitem{Kay1993}
S.~M. Kay, \emph{Fundamentals of statistical signal processing: estimation
  theory}.\hskip 1em plus 0.5em minus 0.4em\relax Upper Saddle River, NJ, USA:
  Prentice-Hall, Inc., 1993.

\bibitem{stoica1998}
P.~Stoica and B.~C. Ng, ``On the {C}ram{\'e}r-rao bound under parametric
  constraints,'' \emph{IEEE Signal Processing Letters}, vol.~5, no.~7, pp.
  177--179, 1998.

\bibitem{yan2017learning}
K.~Yan, L.~Kou, and D.~Zhang, ``Learning domain-invariant subspace using domain
  features and independence maximization,'' \emph{IEEE transactions on
  cybernetics}, vol.~48, no.~1, pp. 288--299, 2017.

\bibitem{spinelle2017field}
L.~Spinelle, M.~Gerboles, M.~G. Villani, M.~Aleixandre, and F.~Bonavitacola,
  ``Field calibration of a cluster of low-cost commercially available sensors
  for air quality monitoring. part b: {NO}, {CO} and {CO$_2$},'' \emph{Sensors
  and Actuators B: Chemical}, vol. 238, pp. 706--715, 2017.

\bibitem{freris10}
N.~Freris, S.~Graham, and P.~Kumar, ``Fundamental limits on synchronizing
  clocks over networks,'' \emph{Automatic Control, IEEE Transactions on}, 2010.

\bibitem{rajanCAMSAP11}
R.~T. Rajan and A.-J. {van der Veen}, ``Joint ranging and clock synchronization
  for a wireless network,'' in \emph{Computational Advances in Multi-Sensor
  Adaptive Processing (CAMSAP), 2011 4th IEEE International Workshop on},
  December 2011, pp. 297 --300.

\bibitem{boydConvexOptimization}
S.~Boyd and L.~Vandenberghe, \emph{{Convex Optimization}}.\hskip 1em plus 0.5em
  minus 0.4em\relax Cambridge University Press, Mar. 2004.

\end{thebibliography}
\bibliographystyle{IEEEtran}

\end{document}